\documentclass[12pt,reqno]{amsart}
\usepackage{amsmath,amsfonts,amsthm,amssymb}
\newcommand\version{July~6, 2026}
\usepackage{amsxtra,graphicx,tikz-cd,hyperref}
\usepackage[mathscr]{euscript}

\newtheorem{theorem}{Theorem}

\newtheorem{lemma}{Lemma}
\newtheorem{corollary}{Corollary}
\newtheorem{assumption}{Assumption}

\theoremstyle{definition}

\theoremstyle{remark}
\newtheorem{remark}{Remark}

\numberwithin{equation}{section}

\newcommand{\Cc}{\mathbb{C}}

\newcommand{\eps}{\epsilon}
\renewcommand{\epsilon}{\varepsilon}

\renewcommand{\phi}{\varphi}
\newcommand{\R}{\mathbb{R}}

\newcommand{\di}{\textnormal{d}}
\newcommand{\mi}{\textnormal{i}}
\newcommand{\eu}{\textnormal{e}}

\DeclareMathOperator{\infspec}{inf\, spec}

\begin{document}

\title[Dyson expansion for form-bounded perturbations ...]{Dyson expansion for form-bounded perturbations, and applications to the polaron problem}

\author{Davide Desio} 
\address{Institute of Science and Technology Austria, Am Campus 1, 3400 Klosterneuburg, Austria}
\email{davide.desio@ist.ac.at}

\author{Robert Seiringer}
\address{Institute of Science and Technology Austria, Am Campus 1, 3400 Klosterneuburg, Austria}
\email{robert.seiringer@ist.ac.at}

\thanks{\copyright\, 2026 by the authors. This paper may be  
reproduced, in
its entirety, for non-commercial purposes.}

\begin{abstract} 
We present an abstract Dyson expansion for perturbations that are merely relatively form-bounded, and apply it to the polaron problem. For a large class of polaron-type models, including the Fr\"ohlich and Nelson models, we prove that the vacuum expectation value of the heat semi-group is a completely monotone function of the square of the total momentum. Consequently, the ground state energy is a concave function of the square of the momentum, a result recently proved for the Fr\"ohlich model in \cite{polzer} using a probabilistic approach via Wiener integrals. 
\end{abstract}

\date{\version}

\maketitle


\section{Introduction and Main results}

The polaron problem concerns the interaction of a charged particle with the  excitations of a medium modeled by a bosonic field. Examples include the Fr\"ohlich model \cite{frohlich} describing  the coupling between an electron and the phonon modes of a polarized crystal, or the Nelson model \cite{nelson} of quantum electrodynamics. 
In this paper, we shall consider a general class of polaron Hamiltonians, formally given by 
\begin{equation}\label{def:ham}
\mathbb{H}(P) =  \left| P - \mathbb{P}_f \right|^2 + \di\Gamma(\omega) + \Phi(v) 
\end{equation}
for $P\in \R^d$ corresponding to the total momentum of the system. The operators \eqref{def:ham} act on the bosonic Fock space
$\mathfrak{F}_{\rm s}(L^2(\R^d))=\Cc\oplus \bigoplus_{N=1}^{\infty}L^2_{\rm s}(\R^{dN})$, where $L^2_{\rm s}(\R^{dN})\subset L^2(\R^{dN})$ denotes the subspace of permutation-invariant functions of $N$ (momentum) variables $k_1,\dots,k_N \in \R^d$. Here, the field energy $\di\Gamma(\omega):=\int_{\R^d}\omega(k)\,a^*_ka_k\,\di k$ is the second quantization of a non-negative Fourier multiplier $\omega$ and $\mathbb{P}_f:=\di\Gamma(k)$ is the field momentum operator. The interaction of the particle with the field modes is described by the  field operator $\Phi(v)= a^\dagger(v) + a(v)$, where $a^\dagger(f)$ and $a(f)$  are, for $f\in L^2(\R^d)$, the usual creation and annihilation operators, satisfying the canonical commutation relations $[a(f),a^\dagger(g)] = \langle f| g\rangle$. The Hamiltonians \eqref{def:ham} result from a restriction of the system to the sector of total momentum $P$, applying the Lee--Low--Pines transformation \cite{LLP53}.

The usual Fr\"ohlich model corresponds to $d=3$, $\omega \equiv 1$, and $v(k) = \sqrt{\alpha/(2\pi^2)} |k|^{-1}$ for some coupling constant $\alpha >0 $. For the Nelson model $\omega(k) = |k|$ (or, more generally, $\omega(k) = \sqrt{|k|^2 + m^2}$ for $m\geq 0$) and $v(k)$ is proportional to $|k|^{-1/2}$, which needs to be suitably cut off at large momentum to avoid an ultraviolet divergence, however. In general, the Hamiltonian \eqref{def:ham} is well-defined as a quadratic form, and bounded from below, under the following assumptions on $v$ and $\omega$.

\begin{assumption}\label{ass1}
We assume that $\omega:\R^d \to \R_+$, and that the interaction $v$ is of the form 
$v= v_1 + v_2$ with
\begin{equation}\label{LYcond}
\int_{\R^d} \frac{|v_1(k)|^2}{\omega(k)} \di k < \infty \quad \text{and} \quad \int_{\R^d} \frac{|v_2(k)|^2}{|k|^2}\left( 1+ \frac 1{\omega(k)}\right) \di k < \infty\,.
\end{equation}
\end{assumption}

The first condition in~\eqref{LYcond} can be viewed as a condition on infrared regularity, while the second concerns the ultraviolet behavior of $v$. 
Under Assumption~\ref{ass1},  the commutator method of Lieb--Yamazaki \cite{LY} shows that $\Phi(v)$ is infinitesimally form-bounded with respect to $|P- \mathbb{P}_f|^2 + \di\Gamma(\omega)$, and hence 
$\mathbb{H}(P)$ defines a semi-bounded self-adjoint operator via its quadratic form \cite[Thm.~X.17]{RS}. For completeness, we shall repeat this argument in Appendix~\ref{appendix:B}.  If $\omega(k)\equiv 1$, the condition \eqref{LYcond} becomes $\int_{\R^d} |v(k)|^2 (1+|k|^2)^{-1} \di k < \infty$, and is satisfied for $v(k) =g |k|^{-\beta}$ for $(d-2)/2<\beta<d/2$ and $g\in\Cc$, for instance. If $\omega(k) = |k|$, on the other hand, $(d-2)/2<\beta<(d-1)/2$ would be admissible. 

To state our main results, we need to introduce some notation. For $n\geq 1$, let $S_{2n}$ denote the set of permutations of $2n$ elements, and let  $\mathcal{W}_{2n}\subset S_{2n}$ be the set of Wick pairings, i.e., those permutations satisfying  
 $\pi(2j-1) < \pi(2j+1)$ for $j\in\{1,\dots,n-1\}$ and $\pi(2j-1) < \pi(2j)$ for $j\in\{1,\dots,n\}$. Permutations in $\mathcal{W}_{2n}$ are in one-to-one correspondence with partitions into sets of cardinality two, i.e.,  $n$ pairings $(\pi(2i-1),\pi(2i))$, $1\leq i \leq n$, of the numbers $\{1,\dots,2n\}$. 
Moreover, we shall denote by $\mathcal{W}_{2n}^0 \subset \mathcal{W}_{2n}$ those permutations that correspond to pairings that interlace, i.e., cannot be decomposed into pairings within disjoint subintervals; explicitly, $\pi \in \mathcal{W}_{2n}$ is an element of $\mathcal{W}_{2n}^0$ if and only if $\sum_{j=1}^\ell (-1)^{\pi^{-1}(j)} < 0$ for all $1\leq \ell < 2n$. 

For $\pi \in \mathcal{W}_{2n}$ and $1\leq j \leq 2n-1$, let 
\begin{equation}\label{def:M}
\mathcal{M}^\pi_j  = \{ 1\leq i \leq n \, : \, \pi(2i-1) \leq j < \pi(2i)\} \subset\{1 ,\dots, n\}
\end{equation}
corresponding to the set of pairs whose connecting line crosses $j$ (which may be empty). For convenience we shall also introduce  $\mathcal{M}^\pi_0 = \mathcal{M}^\pi_{2n} =\emptyset$. 
For $\underline k = (k_1,\dots,k_n) \in \R^{dn}$ and $0\leq j\leq 2n$, let 
\begin{equation}\label{def:b}
\mathscr{E}^{(\pi,j)}_{P}(\underline k) =   \Big| P - \sum\nolimits_{\ell \in \mathcal{M}^\pi_j}  k_\ell \Big|^2 +  \sum\nolimits_{\ell \in \mathcal{M}^\pi_j } \omega(k_\ell)\,.
\end{equation}

We shall denote by $\Delta_k^t$ the simplex 
\begin{equation}\label{def:simplex}
\Delta_{k}^t = \{(t_0,\dots,t_{k})  = \underline t \in \R_+^{k+1}, \ \sum_{i=0}^{k} t_i = t\} \,.
\end{equation}
Finally, we shall use the notation $\Omega = (1,0,0,\dots) \in\mathfrak{F}_{\rm s}(L^2(\R^d))$ for the vacuum vector. 

Our first main result is summarized in the following Theorem. 

\begin{theorem}\label{thm:1}
For $\pi \in \mathcal{W}_{2n}$, the set of Wick pairings defined above, let $\mathscr{E}^{(\pi,j)}_{P}$ be given in \eqref{def:b}. Under Assumption~\ref{ass1}, the following holds:

\noindent (a) For $t>0$, we have the convergent expansion
\begin{equation}\label{id1}
\langle\Omega | \eu^{-t \mathbb{H}(P) } \Omega \rangle = \eu^{-t |P|^2} + \sum_{n\geq 1} \sum_{\pi \in \mathcal{W}_{2n}}  \int_{\Delta_{2n}^t} \int_{\R^{dn} }  \eu^{-\sum_{j=0}^{2n} t_j \mathscr{E}^{(\pi,j)}_{P}(\underline k) }  \prod_{j=1}^n |v(k_j)|^2  \di k_j \,\di\underline t 
\end{equation}
{}where 
the outer integral is over the simplex $\Delta_{2n}^t$ defined in \eqref{def:simplex}. 

\noindent (b) For $t>0$, we further have the renewal equation
\begin{align}\nonumber 
&\langle\Omega | \eu^{-t \mathbb{H}(P) } \Omega \rangle  =  \eu^{-t |P|^2}  \\ & +  \int_0^t  \langle\Omega | \eu^{-(t-s) \mathbb{H}(P) } \Omega \rangle  
\sum_{n\geq 1} \sum_{\pi \in \mathcal{W}^0_{2n}}  \int_{\Delta_{2n-1}^s} \int_{\R^{dn} }  \eu^{-\sum_{j=0}^{2n-1} t_j \mathscr{E}^{(\pi,j)}_{P}(\underline k) }  \prod_{j=1}^n |v(k_j)|^2  \di k_j \,\di\underline t \, \di s
\label{id2}
\end{align}
where $\mathcal{W}^0_{2n} \subset \mathcal{W}_{2n}$ are the interlacing Wick pairings defined above.
\end{theorem}

The identity \eqref{id1} formally follows from a Dyson expansion of $\eu^{- t\mathbb{H}(P)}$ in terms of the interaction $\Phi(v)$, but it is a priori not clear how to make sense of such an expansion in the low-regularity setting considered here, where $\Phi(v)$ is unbounded and merely form-bounded relative to the non-interacting Hamiltonian. An abstract result applicable to this setting will be presented in Section~\ref{sec:dyson}, and will be the starting point in the proof of Theorem~\ref{thm:1}. Part (b) of the theorem is inspired by \cite{polzer}, where a renewal equation for $\langle\Omega | \eu^{-t \mathbb{H}(P) } \Omega \rangle$ was derived in the Fr\"ohlich case, using a closely related probabilistic approach via Wiener integrals.

Theorem~\ref{thm:1} turns out to be very useful in studying the momentum dependence $P \mapsto \langle\Omega | \eu^{-t \mathbb{H}(P) } \Omega \rangle$ for fixed $t>0$. We shall need the following additional assumption on $v$ and $\omega$. Recall that, according to Bernstein's theorem \cite{dono}, a function on $\R_+$ is completely monotone if and only if it is the Laplace transform of a positive measure.

\begin{assumption}\label{ass2}
Assume  
that $\omega$ and $v$ are radial, and that the function
\begin{equation}\label{cma}
|P|^2 \mapsto \int_{\R^d} |v(k)|^2 \eu^{-r |P-k|^2 - s \omega(k)} \di k 
\end{equation}
is completely monotone 
for any $r>0$ and $s>0$. 
\end{assumption}

For $\omega\equiv 1$, this additional assumption is fulfilled for $v(k) = g |k|^{-\beta}$ for any $\beta \geq 0$, 
as one can see by writing $|v(k)|^2$ as an integral over Gaussians, 
\begin{equation}
\Gamma(\beta) |k|^{-2\beta} = \int_0^\infty \eu^{-|k|^2 x } x^{\beta -1 }  \, {\di x}
\end{equation}
for $\beta>0$, 
and noting that the convolution of two Gaussians is again a Gaussian. The same applies to $\omega(k) = \sqrt{|k|^2+m^2}$ for $m\geq 0$, since in this case also $\eu^{-s \omega(k)}$ can be written as an average over Gaussians:
\begin{equation}
\eu^{-s \sqrt{|k|^2 + m^2}} = \int_0^\infty  \frac{\eu^{-x}}{\sqrt{ \pi x}} \eu^{-\frac {s^2}{4x} (|k|^2 + m^2)} \, {\di x}\,.
\end{equation}
 The assumption is thus satisfied for  both the Fr\"ohlich and the Nelson models, if in the latter case the ultraviolet cutoff is chosen in a suitable way to respect the complete positivity of \eqref{cma}, e.g., via a Gaussian factor as $v(k) = g |k|^{-1/2} \eu^{-\eps |k|^2}$ for $\eps>0$. Since the cutoff can be removed, i.e., the limit $\eps\to 0$ taken, after the subtraction of a diverging (but $P$-independent) constant, the following consequently also holds for the renormalized Nelson model.

\begin{corollary}\label{cor}
Under Assumptions~\ref{ass1} and~\ref{ass2}, the following holds:

\noindent (a) $|P|^2 \mapsto \langle\Omega| \eu^{-t \mathbb{H}(P)} \Omega\rangle$ is completely monotone for any $t>0$.

\noindent (b) If $v\not \equiv 0$, then $|P|^2 \mapsto E_0(P) = \infspec \mathbb{H}(P)$  is  concave, and in fact strictly concave on the set of $P$s where $\mathbb{H}(P)$ has a ground state.
\end{corollary}

The existence of a ground state of $\mathbb{H}(P)$ is guaranteed if $E_0(P) < E_{\rm ess}(P)$, where $E_{\rm ess}(P)$ denotes the bottom of the essential spectrum of $\mathbb{H}(P)$. Under additional regularity assumptions on $\omega$, it is known to have the explicit form \cite{Mol06}
\begin{equation}
E_{\rm ess}(P) = \inf_{n\geq 1}  \inf_{k_1,\dots,k_n} \left\{  E_0 \left(P- \sum\nolimits_{j=1}^n k_j \right) + \sum\nolimits_{j=1}^n \omega(k_j) \right\}\,.
\end{equation}
In the case of the Fr\"ohlich model, the infimum is attained at $n=1$ and $k_1=P$, hence $E_{\rm ess}(P) = E_0(0) + 1$ 
for any $P\in \R^d$. 

Since a completely monotone function is log-convex,  the concavity in (b) follows immediately from (a), using that 
\begin{equation}\label{limE}
 E_0(P) = - \lim_{t\to \infty}  \frac 1 t \ln \langle\Omega| \eu^{-t \mathbb{H}(P)} \Omega\rangle \,.
\end{equation}
If $v\equiv 0$, \eqref{limE} holds only whenever $E_0(P) < E_{\rm ess}(P)$, but if $v \not\equiv 0$, it actually holds for any $P\in \R^d$, by a Perron--Frobenius type argument \cite{JMS,KLVW,MiY10} (see also \cite{HH} and references there). 
To obtain strict concavity, the renewal equation of part (b) in Theorem~\ref{thm:1} will be useful, similarly as in \cite{polzer}. 

The concavity of $|P|^2\mapsto E_0(P)$ has played an important role in the recent proof of the validity of the Landau--Pekar formula for the effective mass of the Fr\"ohlich model in the strong coupling limit, see \cite{brooks,brooksS}.

\section{Abstract Dyson Expansion for Form-Bounded Perturbations}\label{sec:dyson}

Let $A\geq 0$ be a self-adjoint operator. Let $B$ be a symmetric quadratic form that is form-bounded with respect to $A$ with relative bound less than $1$, i.e., there exist $a \geq 0$ and $0 < \lambda < 1$ such that 
\begin{equation}\label{assB}
|  B(\psi, \psi)  |  \leq \lambda \left(  \| \sqrt{A} \psi\|  + a \| \psi\|^2 \right)
\end{equation}
for any $\psi$ in the quadratic form domain of $A$ (i.e., the domain of $\sqrt{A}$). Then the semi-bounded quadratic form $\| \sqrt{A}\psi\|^2 + B(\psi,\psi)$ (defined on the form domain of $A$) defines a unique self-adjoint operator \cite[Thm.~X.17]{RS}, which we shall denote by $A+B$. 
Let
\begin{equation}\label{def:C}
C = (A+a)^{-1/2} (A+ B) (A+a)^{-1/2} - A (A+a)^{-1}
\end{equation}
which one readily checks to be a bounded operator satisfying  $\|C\| \leq \lambda <1$. 
It will be convenient to think of $B$  as formally given by $B = \sqrt{A+a} C \sqrt{A+a}$; the precise way to interpret this identity is $A+B + a = \sqrt{A+a}(C+1)\sqrt{A+a}$, which directly follows from the definitions. 

Our goal in this section is to justify the Dyson expansion, formally given by
\begin{equation}\label{abstract:dyson}
\eu^{-t(A+B)} = \eu^{-t A} + \sum_{n\geq 1}(-1)^n \int_{\Delta_n^t} \eu^{-t_0 A} B \eu^{-t_1 A} \cdots B \eu^{-t_n A} \di\underline t
\end{equation}
for $t>0$, where the integral is over the simplex  $\Delta_n^t$ defined in \eqref{def:simplex}. 
However, in the general setting considered here, where $B$ is merely form-bounded relative to $A$, it is not clear how to even make sense of the integrals on the right-hand side. 
Note that even under the stronger assumption that $B$ is operator-bounded relative to $A$, the validity of \eqref{abstract:dyson} is not known, in general; we refer to \cite{rhandi,brendle} for partial results in this direction. 

The following theorem states that under the assumption~\eqref{assB} above, $\eu^{-t(A+B)}$ indeed has a convergent expansion in powers of $B$, but the individual terms in the expansion cannot necessarily be written in the form \eqref{abstract:dyson}. This will be sufficient for our application in the proof of Theorem~\ref{thm:1}, however.

\begin{theorem}\label{thm:2}
Let $A \geq 0$ be a self-adjoint operator,  $B$ a relatively form bounded perturbation satisfying \eqref{assB} for $0\leq \lambda<1$ and $a\geq 0$, and let $A+B$ be defined as above. Then, for $t>0$, there exist bounded operators $D_n$ satisfying
\begin{equation}\label{bound:on:D}
\| D_n\|  \leq \frac {\eu^{3/2}}{\pi   } \sqrt{n+2}\, \lambda^n \eu^{ta}
\end{equation}
such that one has the norm-convergent expansion
\begin{equation}\label{ade}
\eu^{-t (A+B)} = \eu^{-tA} +  \sum_{n\geq 1} D_n \,.
\end{equation}
\end{theorem} 

The proof gives an explicit formula for $D_n$ in terms of a contour integral and the operator $C$ in \eqref{def:C}, see Eq.~\eqref{def:D_n} below. For bounded $B$, one can show that it agrees with the $n$'th term on the right-hand side of \eqref{abstract:dyson}.

\begin{proof}
By absorbing $a$ into $A$, we can set $a=0$ without loss of generality, and assume that $A$ has a bounded inverse. 
We start with the contour integral representation
\begin{equation}
\eu^{-t (A+B)} = \frac {1}{2\pi \mi} \int_{\Gamma_{\gamma,\kappa}} \eu^{- t z} \frac 1 { z - A  - B} \di z
\end{equation}
 for  the contour 
\begin{equation}
\Gamma_{\gamma,\kappa} = \{ z \in \Cc \, : \, \Re z = -\kappa + \gamma | \Im z |  \}
\end{equation}
for some $\gamma > 0$ and  $\kappa> 0$, integrated counter-clockwise (i.e., from top to bottom), see Figure~\ref{fig:path}.
\begin{figure}[t!]
\centering
\begin{tikzpicture}[scale=1.,>=stealth]

  \draw[->] (-4,0) -- (4,0) node[right] {$\mathfrak{Re}z$};
  \draw[->] (0,-3) -- (0,3) node[above] {$\mathfrak{Im}z$};


  \draw[thick, blue] (-2,0) -- (3,1.5) node[right] {$\Gamma_{\gamma,\kappa}$};
  \draw[thick, blue, -<] (-2,0) -- (0.5,0.75);

  \draw[thick, blue] (-2,0) -- (3,-1.5)node[right] {};
  \draw[thick, blue, -<] (3,-1.5) -- (0.5,-0.75);

  \draw[thick, red] (-2,0) -- (-2,3) node[left]{$\mathfrak{Re}z= -\kappa$};
  \draw[thick, red, -<] (-2,0) -- (-2,1.5);
  \draw[thick, red] (-2,-3) -- (-2,0) ;
  \draw[thick, red, -<] (-2,-3) -- (-2,-1.5);
\end{tikzpicture}
\caption{The contour $\Gamma_{\gamma,\kappa}$.}
\label{fig:path}
\end{figure}
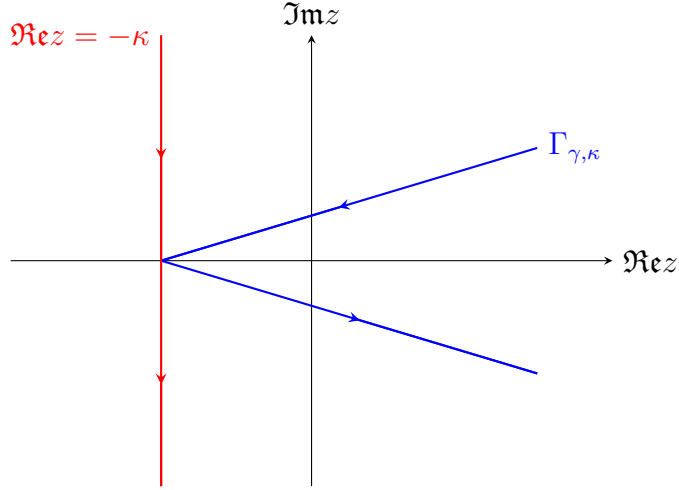
In the following we shall use that the distance between $\Gamma_{\gamma,\kappa}$ and a point on the positive real axis is 
\begin{equation}\label{dist}
\inf_{z\in \Gamma_{\gamma,\kappa} } | z - \lambda | =  \frac{\kappa+\lambda}{\left(1+\gamma^2\right)^{1/2}}\quad  \text{for $\lambda>0$.}
\end{equation}
With $C$ defined as in \eqref{def:C} (with $a=0$), we have the identity 
\begin{equation}
z - A - B =  \sqrt{A} \left( 1-  C  A(z-A)^{-1}  \right)  A^{-1/2}   (z-A) \,.
\end{equation}
Moreover, for $z\in\Gamma_{\gamma,\kappa}$,  \eqref{dist} implies that  
\begin{equation}\label{norm:A}
\left\| A (z-A)^{-1}  \right\| \leq \sup_{\lambda>0} \frac \lambda {|z-\lambda|} \leq  \left( 1+ \gamma^{2} \right)^{1/2}
\end{equation}
and hence $\| CA(z-A)^{-1}  \| < 1$ for $\gamma$ small enough. Under this assumption on $\gamma$, we thus have the convergent Neumann series
\begin{align}\nonumber
\left( z - A - B\right)^{-1}  & = A^{1/2} (z-A)^{-1} \left( 1- CA (z-A)^{-1}  \right)^{-1} A^{-1/2}      \\
& = \frac {A^{1/2}}{z-A} \sum_{n\geq 0} \left( C \frac A{z-A}  \right)^n  A^{-1/2} \,.
\end{align}
By dominated convergence, we can interchange the sum and the integral. 
In particular, we conclude that \eqref{ade} holds 
with
\begin{equation}\label{def:D_n}
D_n =  \frac {1}{2\pi \mi} \int_{\Gamma_{\gamma,\kappa}} \eu^{- t z}   \frac {A^{1/2}} {z-A}  \left(C \frac A{z-A}  \right)^{n-1} C \frac {A^{1/2}} {z-A}  \di z
\end{equation}
for $n\geq 1$. Note that $D_n$ is actually independent of $\gamma$ and $\kappa$. 
Since 
\begin{equation}
\left\| A^{1/2}(z-A)^{-1}  \right\| \leq  \sup_{\lambda>0} \frac {\sqrt \lambda} {|z-\lambda|} \leq  \left( 1+ \gamma^{2} \right)^{1/2} \kappa^{-1/2}
 \end{equation}
for $z \in \Gamma_{\gamma,\kappa}$ (again using \eqref{dist}), the norm of $D_n$ is bounded as 
\begin{equation}
\| D_n\| \leq \frac {1}{2\pi \kappa }  \left( 1+ \gamma^{2} \right)^{(n+1)/2} \| C \|^n   \int_{\Gamma_{\gamma,\kappa}} \eu^{- t \Re z}    |\di z| =  \frac {\eu^{t\kappa}}{\pi \kappa\gamma t }  \left( 1+ \gamma^{2} \right)^{n/2 + 1} \| C \|^n    \,.
\end{equation}
The optimal choice of $\kappa$ and $\gamma$ is thus $\kappa = 1/t$, $\gamma^2 = (n+1)^{-1}$, leading to 
\begin{equation}
\| D_n\| \leq  \frac {\eu}{\pi   } \frac { (n+2)^{(n+2)/2}}{ (n+1)^{(n+1)/2}}  \| C \|^n \leq \frac {\eu^{3/2}}{\pi   } \sqrt{n+2}  \| C \|^n
\end{equation}
uniformly in $t>0$, and thus to the claim in \eqref{bound:on:D} after shifting back by $a$. 
\end{proof}

\begin{remark}
The proof above can be modified to also show that 
\begin{equation}
\left\| \sqrt{A} D_n \sqrt{A} \right\| \leq  \frac {\eu}{\pi  t  } \frac { (n+2)^{(n+2)/2}}{ (n+1)^{(n+1)/2}}  \| C \|^n
\end{equation}
which diverges as $t\to 0$, however. 
\end{remark}

\section{Application to the Polaron Model}

We shall now apply the result of the previous section to the polaron problem. We shall write $\mathbb{H}(P) =  A + B - a$ for suitable (i.e., large enough) $a\geq 0$, 
with
\begin{equation}
A =  \left| P - \mathbb{P}_f\right|^2 + \di\Gamma(\omega) + a
\end{equation}
and $B= \Phi(v) = a(v) + a^\dagger(v)$. As already mentioned in the introduction, the commutator method of Lieb--Yamazaki (recalled in Appendix~\ref{appendix:B}) shows that under Assumption~\ref{ass1},  $B$ is infinitesimally form-bounded with respect to $| P - \mathbb{P}_f| ^2 + \di\Gamma(\omega)$, and in particular $\| C \| < 1$ for   $a$ large enough, where $C=  A^{-1/2} (A+B) A^{-1/2}-1$.

With $\Omega$ the vacuum state, we have $A \Omega = (|P|^2+a)\Omega$. From \eqref{ade} we thus obtain the convergent expansion
\begin{align}\nonumber
\langle\Omega | \eu^{-t \mathbb{H}(P) } \Omega \rangle & = \eu^{t a} \langle\Omega | \eu^{-t (A+B) } \Omega\rangle \\ &= \eu^{-t |P|^2} + \sum_{n\geq 1}  \frac {\eu^{t a}  }{2\pi \mi} \int_{\Gamma_{\gamma,\kappa}}  \frac{\eu^{- t z} }{z-|P|^2-a}  \langle \Omega |    \left( C \frac A{z-A}  \right)^{2n}  \Omega\rangle \,\di z\label{24} \,.
\end{align}
Here we used that only even $n$ in \eqref{ade} contribute, the odd terms have zero vacuum expectation value (since they necessarily involve an unequal number of creation and annihilation operators). 

Let $C^+ = A^{-1/2} a^\dagger(v) A^{-1/2}$ and $C^- = A^{-1/2} a(v) A^{-1/2}$, so that $C = C^+ + C^-$. We have
\begin{equation}\label{sumit}
\langle \Omega |       \left( C \frac A{z-A}  \right)^{2n}  \Omega\rangle = \sum_{\nu \in \mathcal{D}_n} \langle \Omega | \prod_{j=1}^{2n} \left( C^{\nu_j} \frac A{z-A} \right)  \Omega\rangle
\end{equation}
where we identify $C^{+1} \equiv C^+$ and $C^{-1} \equiv C^-$, and  $\mathcal{D}_n$ denotes the set of Dyck paths of length $2n$, i.e., 
\begin{equation}\label{def:Dn}
\mathcal{D}_n = \left\{ \nu \in \{ -1,+1\}^{2n} \, : \, \sum_{j=1}^k \nu_j \leq 0 \  \text{for $1\leq k \leq 2n-1$}, \ \sum_{j=1}^{2n} \nu_j = 0\right\}\,.
\end{equation}
Only these terms contribute to the sum in \eqref{sumit}, since particles can be annihilated only after they have been created, i.e., for a particle annihilated at step $j$, there had to be one created at some step $j'<j$. See Figure~\ref{fig:dyck} for an illustration of a Dyck pack. 

\begin{figure}[ht!]
    \centering
    \begin{tikzpicture}[scale=0.8, >=stealth]

        \draw[->] (-1,0) -- (10,0) node[right] {};
        \draw[->] (0,-1) -- (0,4) node[left] {};

        \draw[thick, blue] 
            (1,0) -- (2,1) 
            -- (3,0) 
            -- (4,1) 
            -- (5,2) 
            -- (6,1) 
            -- (7,2)
            -- (8,1)
            -- (9,0);
        \foreach \x/\y in {1/0, 2/1, 3/0, 4/1, 5/2, 6/1, 7/2, 8/1, 9/0} {
            \filldraw (\x,\y) circle (2pt);
        }
    \end{tikzpicture}
    \caption{Example of a Dyck path of length $2n = 8$, corresponding to $\nu = \{-1,+1,-1,-1,+1,-1,+1,+1\}$. Read from the right,  we represent a creation (+1) with a  line going up in the Dyck path, and an annihilation (-1) with a  line going down. 
    } \label{fig:dyck}
\end{figure}
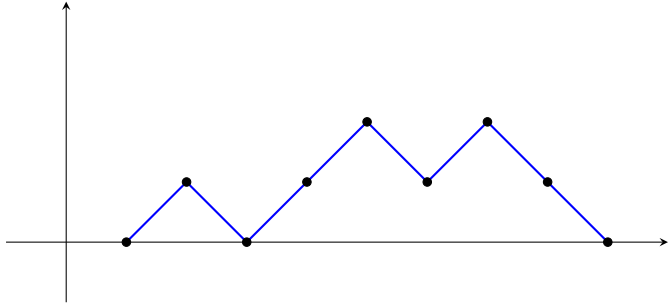

For $\nu\in \mathcal{D}_n$ for some $n\geq 1$, 
let
\begin{equation}\label{def:Lambda}
\Lambda_z(\nu) = \langle \Omega |  \prod_{j=1}^{2n} \left( C^{\nu_j} \frac A{z-A} \right)  \Omega\rangle\,.
\end{equation}
From \eqref{24}--\eqref{sumit}, we thus have 
\begin{equation}
\langle\Omega | \eu^{-t \mathbb{H}(P) } \Omega \rangle   = \eu^{-t |P|^2} +\sum_{\nu \in \mathcal{D}}  \frac {\eu^{t a}  }{2\pi \mi} \int_{\Gamma_{\gamma,\kappa}}  \frac{\eu^{- t z} }{z-|P|^2-a}  \Lambda_z (\nu)\, \di z   \,.\label{2rep}
\end{equation}
To compute the terms on the right-hand side, we need some further notation. 
Recall that $\mathcal{W}_{2n}\subset S_{2n}$ denotes the set of those permutations satisfying the Wick rule,  
i.e., $\pi(2j-1) < \pi(2j+1)$ for $j\in\{1,\dots,n-1\}$ and $\pi(2j-1) < \pi(2j)$ for $j\in\{1,\dots,n\}$.  
Recall also the definition of $\mathscr{E}^{(\pi,j)}_{P}$ in \eqref{def:b}, and let 
\begin{equation}\label{def:S}
\mathscr{S}^{\pi}_{P}(z, \underline k ) 
=\prod_{j=1}^{2n} \left( z - a -  \mathscr{E}^{(\pi,j)}_{P}(\underline k)\right)^{-1} 
\end{equation}
for $\pi \in \mathcal{W}_{2n}$ and $\underline k = (k_1,\dots,k_n)$.

\begin{lemma}\label{key:lem}
Under Assumption~\ref{ass1},
\begin{equation}\label{lem:inte}
\sup_{z\in \Gamma_{\gamma,\kappa}} \int_{\R^{dn} } \left| \mathscr{S}^{\pi}_{P}(z, \underline k )   \right| \prod_{j=1}^n |v(k_j)|^2  \di k_j < \infty
\end{equation}
for any $\pi \in \mathcal{W}_{2n}$ and $\gamma,\kappa>0$. 
\end{lemma}

\begin{proof}
Recall the definition of $\mathcal{M}^\pi_j $ in \eqref{def:M}. 
A little thought reveals that, for any $1\leq j \leq 2n$, $\mathcal{M}^\pi_{j-1}$ and $\mathcal{M}^\pi_j $ differ  by exactly one element, namely the $i \in \{ 1 ,\dots , n\}$ such that $j\in \{ \pi(2i-1), \pi(2i)\}$. We shall denote it by  $i_j$. With the aid of \eqref{dist}, 
%
 one readily finds that
\begin{equation}\label{rf}
 \left|  \left( z - a -  \mathscr{E}^{(\pi,j-1)}_{P}(\underline k)\right) \left( z - a -  \mathscr{E}^{(\pi,j)}_{P}(\underline k)\right) \right|  \geq \frac { (\kappa + a)( \kappa + a + \omega(k_{i_j}) + \frac 12 |k_{i_j}|^2 ) }{ 1+ \gamma^2}
\end{equation}
for $z \in \Gamma_{\gamma,\kappa}$ and  $1\leq j \leq 2n$. Since $\mathcal{M}_0^\pi = \mathcal{M}_{2n}^\pi = \emptyset$ and thus $\mathscr{E}^{\pi,0}_P = \mathscr{E}^{\pi,2n}_P = |P|^2$,
we can write
\begin{equation}
\left| \mathscr{S}^{\pi}_{P}(z, \underline k ) \right|  
=\prod_{j=1}^{2n} \left|  z - a -  \mathscr{E}^{(\pi,j-1)}_{P}(\underline k)\right|^{-1/2}  \left|  z - a -  \mathscr{E}^{(\pi,j)}_{P}(\underline k)\right|^{-1/2} 
\end{equation}
and hence \eqref{rf} yields the bound 
\begin{equation}\label{byc}
\left|\mathscr{S}^{\pi}_{P}(z, \underline k ) \right| \leq \left( \frac { 1+ \gamma^2}{ \kappa + a} \right)^n   \prod_{j=1}^{2n}  \left( \kappa + a + \omega(k_{i_j}) + \frac 12 | k_{i_j}|^2 \right)^{-1/2}  \,.
\end{equation}
By construction, each $k_i$ appears exactly twice in the list $\{ k_{i_1}, \dots, k_{i_{2n}}\}$, hence \eqref{byc}  reads
\begin{equation}
\left| \mathscr{S}^{\pi}_{P}(z, \underline k ) \right| \leq \left( \frac { 1+ \gamma^2}{ \kappa + a} \right)^n   \prod_{i=1}^{n} \frac { 1 }{  \kappa + a + \omega(k_i) + \frac 12 |k_{i}|^2  } \,.
\end{equation}
The claim \eqref{lem:inte} then follows from our assumptions on $v$ and $\omega$ in Assumption~\ref{ass1}. 
\end{proof}

Each Wick pairing $\pi \in \mathcal{W}_{2n}$ corresponds to a Dyck path defined by $\nu_{\pi(j)} = (-1)^j$. We can thus decompose 
$\mathcal{W}_{2n} = \bigcup_{\nu \in \mathcal{D}_n} \mathcal{W}_{2n}^\nu$, where $\mathcal W_{2n}^\nu \subset \mathcal{W}_{2n}$ denotes the set of  those Wick pairings giving rise to the same Dyck path $\nu$.
 With this notation, we have the following identity for $\Lambda_z(\nu)$ defined in \eqref{def:Lambda}.

\begin{lemma}\label{key:lem2}
Under Assumption~\ref{ass1}, we have 
\begin{equation}\label{idLa}
\Lambda_z(\nu) = 
\sum_{\pi \in \mathcal{W}_{2n}^\nu}  \int_{\R^{dn} }   \mathscr{S}^{\pi}_{P}(z, \underline k ) \prod_{j=1}^n |v(k_j)|^2  \di k_j
\end{equation}
for any  $\nu \in \mathcal{D}_n$ and $z \in \Gamma_{\gamma,\kappa}$.
\end{lemma}

Formally, the identity \eqref{idLa} follows from an application of the Wick rule and the pull-through formula $a_k \mathcal{R}(\di \Gamma(\omega), \mathbb{P}_f) =\mathcal{R}(\di \Gamma(\omega)+\omega(k), \mathbb{P}_f+k) a_k$ for functions $\mathcal{R}$ and the operator-valued distribution $a_k$ used to define $a(v)$ by $a(v) = \int a_k \overline{v(k)} \di k$. The proof below can be viewed as a rigorous justification of this procedure.

\begin{proof}
 By analyticity and the uniform integrability established in Lemma~\ref{key:lem}, it suffices to consider the case $\Re z < 0$. Moreover, since both sides of \eqref{idLa} are continuous in $v$ (with respect to the norm defined in Assumption~\ref{ass1}), it suffices to consider the case of suitably nice $v$. We shall in fact first assume that $k\mapsto v(k) \eu^{s\omega (k)}$ is in $L^2(\R^d)$ for any $s>0$. In particular, we work under the  assumption $v\in L^2(\R^d)$, in which case $a(v)$ is effectively a bounded operator in our computations (where the number of particles is bounded by $n$). Moreover, if $\omega$ is unbounded, this assumption additionally requires $v$ to suitably decay. (We could choose it to have compact support, for instance.) 
 For $\Re z < 0$ we can write
\begin{equation}
\frac 1{z-A} = - \int_0^\infty \frac{\eu^{s(z-a)}}{(4\pi s)^{d/2}} \eu^{-s \di\Gamma(\omega)} \int_{\R^d} \eu^{\mi  p\cdot  (P - \mathbb{P}_f)} \eu^{- |p|^2 /(4s)} \, \di p\,  \di s
\end{equation}
and hence 
\begin{align}\nonumber
&\Lambda_z(\nu)  = \langle \Omega |  \prod_{j=1}^{2n} \left( a^{\nu_j}(v)\frac 1{z-A}  \right)  \Omega\rangle \\ & 
=\int_{\R_+^{2n}\times \R^{2nd}}  
\langle \Omega |  \prod_{j=1}^{2n} \left(a^{\nu_j}(v) \eu^{-s_j \di\Gamma(\omega)} \eu^{ - \mi p_j \cdot \mathbb{P}_f} \right)  \Omega\rangle  \prod_{j=1}^{2n} \frac{\eu^{s_j(z-a)}}{(4\pi s_j)^{d/2}}  \eu^{\mi  p_j \cdot  P} \eu^{-|p_j|^2/(4s_j)}  \di s_j\, \di p_j
 \label{40}
\end{align}
where we identify $a^-(v)= a(v)$ and $a^+(v) = a^\dagger(v)$. 
Moreover, since
\begin{equation}
\eu^{-s\di\Gamma(\omega)} \eu^{-\mi p \cdot \mathbb{P}_f} a^\nu(v) =  a^\nu(\eu^{-s\nu\omega}\eu^{-\mi p\,\cdot\,}v)\eu^{-s\di\Gamma(\omega)} \eu^{-\mi p \cdot \mathbb{P}_f}
\end{equation}
for $\nu \in \{+,-\}$, we have 
\begin{equation}
\langle \Omega |  \prod_{j=1}^{2n} \left(a^{\nu_j}(v) \eu^{-s_j \di\Gamma(\omega)} \eu^{ - \mi p_j \cdot \mathbb{P}_f}  \right)  \Omega\rangle 
= \langle \Omega |  \prod_{j=1}^{2n}  a^{\nu_j}(v_j)   \Omega\rangle 
\end{equation}
with $v_1=v$ and 
\begin{equation}
v_j(k) = v(k) \eu^{- \omega(k)   \nu_j \sum_{\ell = 1}^{j-1} s_\ell  } \eu^{-\mi  k \cdot \sum_{\ell = 1}^{j-1} p_\ell }
\end{equation}
for $j\geq 2$. 
An application of Wick's rule further yields
\begin{align}\nonumber
\langle \Omega |  \prod_{j=1}^{2n}  a^{\nu_j}(v_j)   \Omega\rangle   & = \sum_{\pi \in \mathcal{W}_{2n}} \prod_{i=1}^n \langle \Omega |  a^{\nu_{\pi(2i-1)}}(v_{\pi(2i-1)})  a^{\nu_{\pi(2i)}}(v_{\pi(2i)})  \Omega\rangle  \\ \nonumber
&  = \sum_{\pi \in \mathcal{W}_{2n}} \prod_{i=1}^n \langle v_{\pi(2i-1)} | v_{\pi(2i)} \rangle \delta_{\nu_{\pi(2i-1)},-1} \delta_{\nu_{\pi(2i)},+1} \\ \nonumber
 & =  \sum_{\pi \in \mathcal{W}_{2n}^\nu } \prod_{i=1}^n \langle v_{\pi(2i-1)} | v_{\pi(2i)} \rangle 
\\ & =  \sum_{\pi \in \mathcal{W}_{2n}^\nu } \prod_{i=1}^n \int_{\R^d}  
|v(k)|^2  
\eu^{ - \omega(k)    \sum_{\ell = \pi(2i-1)}^{\pi(2i)-1} s_\ell  }
  \eu^{-\mi  k \cdot \sum_{\ell = \pi(2i-1)}^{\pi(2i)-1} p_\ell }
 \, \di k \,.
\end{align}
Inserting this in \eqref{40} above and interchanging the order of integrations then gives the desired formula \eqref{idLa} under the stated assumptions on $v$, i.e., under the assumption that   $k\mapsto v(k) \eu^{s\omega (k)}$ is in $L^2(\R^d)$ for any $s>0$. The set of such $v$'s can easily be seen to be dense in the set defined by Assumption~\ref{ass1} for the appropriate norm, hence the general case follows by continuity of both sides of \eqref{idLa},  as already mentioned in the beginning of the proof. 
\end{proof}

We now have all the prerequisites to give the proof of our main results.

\section{Proof of the Main Results}

\subsection{Proof of Theorem~\ref{thm:1}(a)} \label{proof1}
The starting point is \eqref{2rep}, where we insert the identity \eqref{idLa} from Lemma~\ref{key:lem2}. Because of the uniform integrability established in Lemma~\ref{key:lem}, 
we can interchange the integration over $k_1,\dots,k_n$ with the one over $z$, and do the latter first. 
One easily checks that for positive numbers $b_j$ one has 
\begin{equation}
 \frac {1}{2\pi \mi} \int_{\Gamma_{\gamma,\kappa}}  \eu^{- t z} \prod_{j=0}^{2n} \frac{1}{z-b_j} \di z =    \int_{\Delta_{2n}^t} \eu^{-\sum_{j=0}^{2n} t_j b_j } \di\underline t\,.
\end{equation}
(To see this, send $\gamma\to 0$, and write $(z-b_j)^{-1} =- \int_0^\infty \eu^{s(z-b_j)}\di s $ for $ z\in \Gamma_{0,\kappa} =\{ -\kappa + \mi t \, : \, t\in\R\}$, as sketched by the red contour in Figure~\ref{fig:path}.) In particular, for $ \mathscr{S}^{\pi}_{P}$ in \eqref{def:S} 
\begin{equation}\label{46}
  \frac {\eu^{ta}}{2\pi \mi} \int_{\Gamma_{\gamma,\kappa}}  \frac{\eu^{- t z}}{z-|P|^2-a}     \mathscr{S}^{\pi}_{P}(z, \underline k) \di z = \int_{\Delta_{2n}^t} \eu^{-\sum_{j=0}^{2n} t_j \mathscr{E}^{(\pi,j)}_{P}(\underline k) } \di\underline t 
\end{equation}
using again that $\mathscr{E}^{(\pi,0)}_P = |P|^2$. 
By positivity of the integrand, we can exchange the integration over $\underline t$ with the one over $\underline k$, thus arriving at the desired expansion \eqref{id1}. \hfill\qed

\subsection{Proof of Corollary~\ref{cor}(a)}
By Bernstein's Theorem, Assumption~\ref{ass2} implies that for any $r,s>0$ there exists a measure on $\di\mu_{r,s}$ on $\R_+$ such that 
\begin{equation}\label{murs}
\int_{\R^d} \eu^{-r |P-k|^2 - s \omega(k)}  |v(k)|^2 \di k  = \int_0^\infty \eu^{-\lambda |P|^2} \di\mu_{r,s}(\lambda)\,.
\end{equation}
We claim that this implies, for any $n\geq 1$, 
\begin{equation}\label{murs2}
\int_{\R^{dn}}    \eu^{- \mathscr{Q}(P,\underline k)}  \prod_{j=1}^n \eu^{- s_j \omega(k_j)}  |v(k_j)|^2 \di k_j  = \int_0^\infty \eu^{-\lambda |P|^2} \di\mu_{\mathscr{Q},\underline s}(\lambda)
\end{equation}
for a suitable measure $\di\mu_{\mathscr{Q},\underline s}$, for any set of positive parameters $\underline s = (s_1,\dots,s_n) \in \R_+^n$ and any positive quadratic form $\mathscr{Q}$ on $\R^{d(n+1)}$ that is invariant under simultaneous rotations of the $n+1$ variables,  i.e., is of the form $\mathscr{Q}(k_0,\underline k) = \sum_{i,j=0}^n \mathscr{Q}_{ij} k_i \cdot k_j$ for a positive symmetric matrix with coefficients $\mathscr{Q}_{ij}$. 

In the case $n=1$, this follows readily from \eqref{murs}, since $\mathscr{Q}(P,k_1)$ can be written as $\mathscr{Q}_{11}  | k_1 + \mathscr{Q}_{01} \mathscr{Q}_{11}^{-1} P|^2 + (\mathscr{Q}_{00} - \mathscr{Q}_{01}^2 \mathscr{Q}_{11}^{-1}) |P|^2$. 

For general $n$, we can proceed by induction. Singling out the last integration variable $k_n$, we can write 
\begin{equation}
\mathscr{Q}(k_0,\underline k) = \mathscr{Q}_{nn} \Big| k_n +  \mathscr{Q}_{nn}^{-1} \sum_{j=0}^{n-1} \mathscr{Q}_{jn} k_j  \Big|^2 + \mathscr{Q}'(k_0,k_1,\dots,k_{n-1})
\end{equation}
with $\mathscr{Q}'$  a rotation-invariant non-negative quadratic form on $\R^{d n}$, given by 
\begin{equation}
 \mathscr{Q}'_{ij} = \mathscr{Q}_{ij}  -\frac 1{\mathscr{Q}_{nn}} \mathscr{Q}_{in} \mathscr{Q}_{nj} \quad \text{for $0\leq i,j\leq n-1$.}
\end{equation}
Applying \eqref{murs} to the integration over $k_n$, we are left with an integral over similar problems for $n-1$ variables, with quadratic expressions $\lambda | \mathscr{Q}_{0n} P +  \sum_{j=1}^{n-1} \mathscr{Q}_{jn} k_j |^2 + \mathscr{Q}'(P,k_1,\dots,k_{n-1})$ for $\lambda \geq 0$. The claim \eqref{murs2} thus follows by induction.

By positivity of the integrand, we can first carry out the integration over $\underline k$ for fixed $\underline t$ in the various terms in \eqref{id1}. For given $\underline t=(t_0,\dots,t_{2n}) \in \R_+^{2n+1}$ and $\pi \in \mathcal{W}_{2n}$, the function $\sum_{j=0}^{2n} t_j \mathscr{E}^{(\pi,j)}_{P}(\underline k)  $ is of the form 
\begin{equation}
\sum_{j=0}^{2n} t_j \mathscr{E}^{(\pi,j)}_{P}(\underline k)  = \mathscr{Q}(P,\underline k) + \sum_{j=1}^n s_j \omega(k_j)
\end{equation}
for a rotation-invariant positive quadratic form $\mathscr{Q}$ and positive parameters $s_j$. The claimed complete monotonicity is thus an immediate consequence of \eqref{murs2}. \hfill\qed

\subsection{Proof of Theorem~\ref{thm:1}(b)} 

The starting point is again \eqref{2rep}, but we shall change back the order of integration and summation (which is allowed for small enough $\gamma$, as shown in the proof of Theorem~\ref{thm:2} in Section~\ref{sec:dyson}). That is, we have
\begin{equation}
\langle\Omega | \eu^{-t \mathbb{H}(P) } \Omega \rangle   = \eu^{-t |P|^2} +  \frac {\eu^{t a}  }{2\pi \mi} \int_{\Gamma_{\gamma,\kappa}}  \frac{\eu^{- t z} }{z-|P|^2-a}  \sum_{\nu \in \mathcal{D}} \Lambda_z (\nu)\, \di z   \,.\label{2repa}
\end{equation}
Recall the definition of the Dyck paths $\mathcal{D}_n$ in \eqref{def:Dn}, and the one of $\Lambda_z$ in \eqref{def:Lambda}. 
Given $\nu \in \mathcal{D}_{n}$ and $\nu' \in \mathcal{D}_{n'}$, we can concatenate them to $\{\nu,\nu'\} \in \mathcal{D}_{n+n'}$. 
We have
\begin{equation}
\Lambda_z(\{\nu,\nu'\}) = \Lambda_z(\nu)\Lambda_z(\nu') 
\end{equation}
and, if we set $\Lambda_z(\emptyset) = 1$, this continues to hold for $\nu=\emptyset$. Every $\nu\in \mathcal{D}_n$ can be uniquely written as $\nu=\{ \tilde\nu,\nu_0\}$ where $\nu_0 \in \mathcal{D}_m^0$ for some $1\leq m\leq n$ and $\tilde\nu \in \mathcal{D}_{n-m}$ (being possibly empty), where $\mathcal{D}_n^0 \subset\mathcal{D}_n$ denotes those Dyck paths that return to the origin only at the end, i.e., satisfy $\sum_{j=1}^k \nu_j <0$ for all $1\leq k \leq 2n-1$.
Hence, with $\mathcal{D} = \bigcup_{n\geq 1} \mathcal{D}_n$ (and similarly for $\mathcal{D}^0$) \begin{equation}
 \sum_{\nu \in \mathcal{D}} \Lambda_z(\nu) = \left( 1+   \sum_{\tilde\nu \in \mathcal{D} } \Lambda_z(\tilde\nu) \right)  \sum_{\nu_0 \in \mathcal{D}^0} \Lambda_z(\nu_0)\,.
\end{equation}

For two functions $f_1$ and $f_2$ analytic away from the positive half-line (and satisfying suitable growth bounds), we have
\begin{align}\nonumber
& \frac 1{2\pi \mi} \int_{\Gamma_{\gamma,\kappa}} \eu^{-t z} f_1(z) f_2(z) \di z \\ & = \int_0^t  \left( \frac 1{2\pi \mi} \int_{\Gamma_{\gamma,\kappa}} \eu^{-(t-s) z} f_1(z)  \di z \right) \left( \frac 1{2\pi \mi} \int_{\Gamma_{\gamma,\kappa}} \eu^{-s z'} f_2(z')  \di z' \right) \di s \,.
\end{align}
We shall apply this with $f_1(z) = (z-|P|^2-a)^{-1}(1 + \sum_{\nu\in\mathcal{D}}\Lambda_z(\nu))$ and $f_2(z) = \sum_{\nu \in \mathcal{D}^0} \Lambda_z(\nu)$. Using also \eqref{2repa} with $t$ replaced by $s$, we find
\begin{equation}
\langle\Omega | \eu^{-t \mathbb{H}(P) } \Omega \rangle   = \eu^{-t |P|^2} +  \int_0^t  \langle\Omega | \eu^{-(t-s) \mathbb{H}(P) } \Omega \rangle  
 \frac {\eu^{sa}}{2\pi \mi} \int_{\Gamma_{\gamma,\kappa}}  \eu^{- s z} \sum_{\nu \in \mathcal{D}^0} \Lambda_z(\nu)\, \di z \, \di s   \,.
\end{equation}
Lemma~\ref{key:lem2} implies that
\begin{equation}
\sum_{\nu\in\mathcal{D}^0_n} \Lambda_z(\nu) = 
 \sum_{\pi \in \mathcal{W}_{2n}^0}  \int_{\R^{dn} }  \mathscr{S}^\pi_P(z, \underline k) \prod_{j=1}^n |v(k_j)|^2  \di k_j
\end{equation}
where $\mathcal{W}_{2n}^0 \subset \mathcal{W}_{2n}$ are those Wick pairings such that $j\mapsto (-1)^{\pi^{-1}(j)}$ is in $\mathcal{D}^0_{n}$, which are precisely the interlacing Wick pairings defined in the introduction. 
The desired identity \eqref{id2} then follows by proceeding as in Subsection~\ref{proof1} above (compare with \eqref{46}). \hfill\qed

\subsection{Proof of Corollary~\ref{cor}(b)}
As already mentioned in the introduction, concavity of $|P|^2\mapsto E_0(P)$ is an immediate consequence of part (a) in combination with \eqref{limE}. 

For $v\not \equiv 0$ a simple variational argument shows that $E_0(P)<|P|^2$ for all $P\in\R^d$. If we assume that $\mathbb{H}(P)$ has a ground state $\Phi_P$, then
\begin{equation}
\lim_{t\to \infty} \eu^{t E_0(P)} \langle \Omega | \eu^{-t \mathbb{H}(P)} \Omega\rangle = | \langle \Omega | \Phi_P \rangle |^2 > 0
\end{equation}
where the strict positivity follows from the fact that $\eu^{-t \mathbb{H}(P)}$ is positivity improving on a suitable cone containing the vacuum \cite{MiY10}. (The statement is proved in \cite{MiY10} for the Fr\"ohlich model, but the method of proof applies to the general case considered here.)  
We can thus take $t\to \infty$ in \eqref{id2} and obtain, using dominated convergence,
\begin{equation}
1 =  \int_0^\infty\eu^{s E_0(P)}   
\sum_{n\geq 1} \sum_{\pi \in \mathcal{W}^0_{2n}}  \int_{\Delta_{2n-1}^s} \int_{\R^{dn} }  \eu^{-\sum_{j=0}^{2n-1} t_j \mathscr{E}^{(\pi,j)}_{P}(\underline k) }  \prod_{j=1}^n |v(k_j)|^2  \di k_j \,\di\underline t  \, \di s  \,.
\end{equation}
From the complete monotonicity of the terms multiplying $\eu^{s E_0(P)}$ in the integrand, it is then easy to see that H\"older's inequality implies that  $|P|^2\mapsto E_0(P)$ is strictly concave (as  argued in \cite[Cor.~8]{polzer}), using again Bernstein's theorem.  \hfill\qed

\appendix 
\section{The Lieb--Yamazaki Argument}\label{appendix:B}

In this section, we shall recall the Lieb--Yamazaki commutator method introduced in \cite[Sect.~2]{LY} to prove that if $v$ and $\omega$ satisfy Assumption~\ref{ass1}, $\Phi(v)$ is infinitesimally form bounded with respect to the non-interacting Hamiltonian $|P-\mathbb{P}_f|^2 + \di\Gamma(\omega)$.
On the one hand, a simple Cauchy--Schwarz inequality shows that 
\begin{equation}
\pm \Phi(v) \leq \eps \di\Gamma(\omega) + \eps^{-1} \int_{\R^d} \frac {|v(k)|^2}{\omega(k)} \di k
\end{equation}
for any $\eps>0$, which we can apply to $v=v_1$. 
On the other hand, using that
\begin{equation}
\Phi(v) = \mi [ P - \mathbb{P}_f , \Phi(w) ] 
\end{equation}
for $w(k) =\mi k v(k) /|k|^2$ (where an inner product between the vectors $P - \mathbb{P}_f$ and $w$ is understood on the right-hand side), a Cauchy--Schwarz inequality implies that
\begin{equation}
\pm \Phi(v) \leq \eps |P- \mathbb{P}_f|^2 + \eps^{-1} |\Phi(w)|^2\,.
\end{equation}
We can further bound
\begin{equation}
\Phi(w)^2  \leq 4 a^\dagger(w) a(w) + 2 \| w\|^2 \leq \lambda \di\Gamma(\omega) + 2 \int_{\R^d} \frac{ |v(k)|^2}{|k|^2} \di k
\end{equation}
for $\lambda = 4 \int |w(k)|^2 / \omega(k) \di k = 4 \int |v(k)|^2/(|k|^2 \omega(k)) \di k$. We shall apply this bound to $v_2$. By suitably adjusting how to split $v$ into the two components $v = v_1 + v_2$, we can assume that 
$\lambda$ is as small as desired, which implies the claimed infinitesimal bound relative to $|P-\mathbb{P}_f|^2 + \di\Gamma(\omega)$.



\begin{thebibliography}{19}

\bibitem{brendle} S. Brendle, {\it  On the asymptotic behavior of perturbed strongly continuous semigroups}, 
Math. Nachr. {\bf 226}, 35 (2001).

\bibitem{brooks} M. Brooks, {\it Proof of the Landau--Pekar Formula for the effective Mass of the Polaron at strong coupling}, preprint arXiv:2409.08835 

\bibitem{brooksS} M. Brooks, R Seiringer, {\it  The Fr\"ohlich polaron at strong coupling: Part II---Energy-momentum relation and effective mass}, 
Publ. Math. IH\'ES {\bf 140},  271 (2024).

\bibitem{dono} W.F. Donoghue, {\it Monotone Matrix Functions and Analytic Continuation}, Springer (1974). 

\bibitem{frohlich} H. Fr\"ohlich, {\it Theory of electrical breakdown in ionic crystals}, Proc. Roy. Soc. London, Ser. A {\bf 160}, 230 (1937).

\bibitem{HH} B. Hinrichs and F. Hiroshima, {\it On the ergodicity of renormalized translation-invariant Nelson-type semigroups}, preprint arXiv:2412.09708

\bibitem{JMS}
J. L\H orinczi, R.A. Minlos, and H. Spohn, 
{\it  The Infrared Behaviour in Nelson's Model of a Quantum Particle Coupled to a Massless Scalar Field}, 
Annales H. Poincar{\'e} {\bf 3}, 269 (2002).

\bibitem{KLVW} M. Keller, D. Lenz, H. Vogt, and R. Wojciechowski, {\it Note on basic features of large time behaviour of heat kernels}, 
J. reine angew. Math. {\bf 708}, 73--95 (2015).

\bibitem{LLP53}
T.D. Lee, F.E. Low, and D. Pines, {\it The Motion of Slow Electrons in a Polar Crystal}, Phys. Rev. {\bf 90}, 297 (1953).

\bibitem{LY}
E.H. Lieb and K. Yamazaki, {\it Ground-State Energy and Effective Mass of the Polaron}, Phys. Rev. {\bf 111}, 728 (1958).

\bibitem{MiY10}
T. Miyao, {\it Nondegeneracy of ground states in nonrelativistic quantum field theory}, J. Oper. Theory
{\bf 64}, 207 (2010).

\bibitem{Mol06}
J. Schach M\o ller, {\it The Polaron revisited},  Rev. Math. Phys. {\bf 18},  485 (2006).

\bibitem{nelson} E. Nelson, {\it Interaction of nonrelativistic particles with a quantized scalar field}, J. Math. Phys. {\bf 5}, 1190 (1964).

\bibitem{polzer} S. Polzer, {\it Renewal approach for the energy-momentum relation of the Fr\"ohlich polaron}, 
Lett. Math. Phys. {\bf 113}, 90 (2023).

\bibitem{rhandi}
A. Rhandi, {\it Dyson--Phillips expansion and unbounded perturbations of linear $c_0$-semigroups}, 
J. Comp. Appl. Math. {\bf 44}, 339 (1992).

\bibitem{RS}
M. Reed and B. Simon, {\it Methods of Modern Mathematical Physics II: Fourier Analysis, Self-Adjointness}, Academic Press (1975). 

\end{thebibliography}
\end{document}